\documentclass[12pt]{article}
\usepackage{amsfonts,epsfig,latexsym,amscd,amsmath,theorem,mathrsfs}

\DeclareMathOperator{\grad}{grad}
\DeclareMathOperator{\rank}{rank}
\DeclareMathOperator{\supp}{supp}

\newcommand{\R}{{\mathbb{R}}}
\newcommand{\Z}{{\mathbb{Z}}}

\newcommand{\beq}{\begin{equation}}
\newcommand{\eeq}{\end{equation}}
\newcommand{\bea}{\begin{eqnarray}}
\newcommand{\eea}{\end{eqnarray}}
\newcommand{\ben}{\begin{eqnarray*}}
\newcommand{\een}{\end{eqnarray*}}
\newcommand{\ra}{\rightarrow}

\newcommand{\cd}{\partial}

\newcommand{\less}{\backslash}

\def \d{\mathrm{d}}
\newcommand{\dstar}{\delta}
\newcommand{\ip}[1]{\langle #1 \rangle}
\newcommand{\ignore}[1]{}

\newcommand{\ee}{\mathscr{E}}
\newcommand{\hh}{\mathscr{H}}
\newcommand{\lll}{\mathscr{L}}
\renewcommand{\aa}{\mathscr{A}}

\newcommand{\vol}{{\rm vol}}

\newcommand{\eps}{\epsilon}
\renewcommand{\phi}{\varphi}

\theoremstyle{plain}
\newtheorem{thm}{Theorem}
\newtheorem{lemma}[thm]{Lemma}
\newtheorem{prop}[thm]{Proposition}

{\theorembodyfont{\rmfamily}
\newtheorem{defn}[thm]{Definition}
\newtheorem{remark}[thm]{Remark}

}

\newcommand{\news}{\setcounter{equation}{0}}
\newenvironment{proof}{\noindent{\it Proof:\, }}{\hfill$\Box$\vspace*{0.5cm}
}

\begin{document}

\title{Compactons and semi-compactons in the extreme baby Skyrme model}
\author{
J.M. Speight\thanks{E-mail: {\tt speight@maths.leeds.ac.uk}}\\
School of Mathematics, University of Leeds\\
Leeds LS2 9JT, England}

\date{}
\maketitle

\begin{abstract}
The static baby Skyrme model is investigated in the extreme limit 
where the energy functional contains only the potential and Skyrme terms,
but not the Dirichlet energy term. It is shown that the model with
potential $V=\frac12(1+\phi_3)^2$ possesses solutions with extremely
unusual localization properties, which we call semi-compactons. These
minimize energy in the degree 1 homotopy class, have support contained
in a semi-infinite rectangular strip, and decay along the length of the strip
as $x^{-\log x}$. By gluing together several semi-compactons, it is
shown that
every homotopy class has linearly stable solutions of arbitrarily high,
but quantized, energy. For various other choices of potential,
compactons are constructed with support in a closed disk, or in a
closed annulus. In the latter case, one can construct higher
winding compactons, and complicated superpositions in which several
closed string-like compactons are nested within one another. The
constructions make heavy use of the invariance of the
model under area-preserving diffeomorphisms, and of a topological
lower energy bound, both of which are established in a general
geometric setting. All the solutions presented are classical, that is,
they are (at least) twice continuously differentiable and satisfy
the Euler-Lagrange equation of the model everywhere.
\end{abstract}

\maketitle

\section{Introduction}
\news

Solitons are stable, spatially localized solutions of nonlinear 
field theories. Ordinarily, ``spatially localized'' means that the
field $\phi(t,x)$ approaches some constant vacuum value 
$\phi_0$
asymptotically  
as $|x|\ra\infty$, usually exponentially in $|x|$
(e.g.\ KdV and sine-Gordon solitons, abelian Higgs vortices) 
sometimes as a power law $|x|^{-p}$ (e.g.\ sigma model lumps, instantons).
However, there are some systems where the solitons' spatial localization
is much more severe: $\phi(x)=\phi_0$ exactly outside some compact region
of space. Since they have compact support, these solitons are called
{\em compactons} in the literature. They were first discovered 
in generalized KdV equations \cite{roshym}, then in nonlinear Klein-Gordon
models 
with  W-shaped potentials 
\cite{aro}, that is, potentials with two degenerate vacua at each of which the
potential has a V-shaped singularity. 
All these compactons live in one (spatial) dimension.
Moving to two dimensions,
 compactons have been constructed for the baby-Skyrme model,
with energy density 
\beq
\ee=\frac12 \lambda \cd_i\phi\cdot\cd_i\phi +
\frac14(\cd_i\phi\times\cd_j\phi)\cdot(\cd_i\phi\times\cd_j\phi)
+V(\phi)
\eeq
again in the case where $V(\phi)$ has a V-shaped singularity at the vacuum,
for example $V=\sqrt{1-\phi_3}$ \cite{ada2}.  Here 
$\phi:\R^2\ra S^2\subset\R^3$ and $\lambda$ is a positive 
constant. 
Since $V$ is singular, it is not surprising that the compactons
are singular, and this makes their interpretation somewhat problematic.
In what sense, precisely, are they solutions of the model?

In both \cite{aro} and \cite{ada2}, the singularity of $V$ can
be interpreted as $V$ having infinite second derivative at the 
vacuum, so that the ``mesons'' of the theory (propagating small perturbations
about the vacuum) have infinite mass. An alternative mechanism to give
the mesons infinite mass, without making $V$ singular,
 is to take the limit $\lambda\ra 0$ in the 
above model \cite{gis,ada1}. 
This limit, variously called the pure, restricted or
(as in this paper) extreme baby-Skyrme model, has the interesting property 
of being invariant under all area-preserving diffeomorphisms of the spatial
plane, and is claimed to have applications in condensed matter physics 
\cite{gis}. For various choices of (continuously differentiable) potential
$V$, it has been found to support compactons.
One problem with all the compactons found in \cite{gis}, and
some of those found in \cite{ada1}, is that they are not even once 
continuously differentiable, so again it is not clear in what sense they are
solutions of the 
model.
Certainly they are not classical solutions of the 
field equation, which is a second order nonlinear PDE.\,  
They may be solutions in
the weaker sense that they locally extremize the energy functional,
but to make precise sense of this is rather technical: 
given that the fields themselves are not $C^1$, what should be the allowed
space of variations (usually taken to be $C^1$ with compact support)?
Since the compactons in \cite{ada1} saturate a topological lower energy bound,
it is likely that a precise formulation of their status as solutions is 
possible. The compactons in \cite{gis} are more problematic.

In this paper, we begin by
analyzing the extreme baby Skyrme model in a rather general
geometric
setting, taking physical space to be any orientable two-manifold $M$
and target space to be any compact Riemann surface $N$ (so the primary case
of interest is $M=\R^2$ and $N=S^2\subset \R^3$). The potential will be taken
to be of the form $V=\frac12 U^2$ where $U$ is a non-negative $C^1$
function on $N$ with isolated zeros.
By a {solution} of the model we will strictly mean a
twice continuously differentiable map $M\ra N$ satisfying the
Euler-Langrange equation for $E$ everywhere. 
 In this setting, we prove a topological
lower energy bound, saturated by solutions of a first
order ``Bogomol'nyi'' equation.
 Solutions of this equation have a natural
interpretation as area-preserving maps from (part of) $M$ to
(almost all) $N$, with respect to a deformed area form on $N$ (determined
by $U$). We show that all Bogomol'nyi solutions are solutions of the field
equation and conversely (on $M=\R^2$) that all solutions of the 
field equation are (piecewise) Bogomol'nyi. The
 bound is a generalization of various special cases discovered
previously \cite{ada1,izq,pie,war1}, and our main contribution here is to
place these results within a geometric framework, and give a geometric 
interpretation of the Bogomol'nyi equation.

We then consider the
specific case $M=\R^2$, $N=S^2$, $U=1+\phi_3$ in detail. 
By exploiting the model's symmetry under area-preserving diffeomorphisms, 
we construct
a degree $1$ solution of this model with extremely unusual
localization properties,
which we call a {\em semi-compacton}.
This solution is constant outside a semi-infinite
rectangular strip, decays like $x^{-\log x}$ along the length of the
strip, and minimizes energy within its homotopy class. By gluing
several (anti-) semi-compactons together, we show that in every homotopy
class the model has solutions of arbitrarily high (but quantized)
energy, all of which are at least marginally stable (in particular,
they are {\em not} saddle points). We also prove that the critical set
of any solution of the model can have no bounded connected components so,
in particular, solutions can never have isolated critical points. 
We compare
our results with those of Adam et al \cite{ada1}, who construct
exponentially localized fields in this model, clarifying precisely when
their fields are solutions in the strong sense used here.

We go on to consider various cases where $U$ is not $C^1$ but $V=\frac12 U^2$
still is, which is enough for the critical parts of the general theory
to survive,
giving a necessary condition on $U$ for the existence of
compactons. In the case $U=(1+\phi_3)^\alpha$, $\frac12\leq\alpha<1$,
we give a geometric construction of the compactons obtained in \cite{ada1},
and show how their key qualitative features (e.g.\ energy and area)
can be found without solving any equations.
In the case $U=(1+\phi_3)^\alpha(1-\phi_3)^\beta$, 
where $\frac12\leq\alpha,\beta<1$, we construct annular (or closed string-like)
compacton solutions which minimize energy in their homotopy class, generalizing
results in \cite{ada1} (which correspond to the degenerate case where
the annulus is a punctured disk).

In the final section, we consider the model on a compact domain, with $U=0$,
showing that generically the model has no nontrivial solutions at all. We 
conclude by suggesting some interesting open questions concerning the 
dynamics of semi-compactons.

\section{The Bogomol'nyi argument}
\news

It is convenient to place the extreme
baby Skyrme model within a more general geometric
framework. The model has a single scalar field $\phi:M\ra N$ 
where $(M,g)$ is an oriented Riemannian two-manifold, 
representing
physical space, and $(N,h,J)$ is a compact Riemann surface
(with metric $h$ and almost complex structure $J$), the target space. 
Denote
by $\omega=h(J\cdot,\cdot)$ the K\"ahler form (or area form) on $N$. 
Let $U:N\ra [0,\infty)$ be
a $C^1$ function with isolated zeros, the vacua of the theory.
In this section, we will take $M$ either to be compact, or to be Euclidean
$\R^2$. In the latter case (which is of most direct physical interest)
 we impose the boundary
condition
\beq\label{bc}
\phi(x)\ra\phi_0\in U^{-1}(0)\qquad \mbox{as $|x|\ra\infty$}
\eeq
sufficiently fast that $\int_M\phi^*\omega$ converges. 
Throughout, $\phi$ is assumed to be at least $C^2$.
The energy functional of the model is
\bea\label{class}
E(\phi)&=&\frac12\int_M|\phi^*\omega|^2\vol_M+\frac12\int_M(U\circ\phi)^2\vol_M
\\
&=&\frac12\|\phi^*\omega\|^2+\frac12\|U\circ\phi\|^2
=E_4(\phi)+E_0(\phi)\eea
where $\vol_M$ denotes the volume form on $M$, and we have introduced
the notation $\|\cdot\|$ for the $L^2$ norm of a function, or form, on
$M$. It will be convenient to denote the associated $L^2$ inner product
by $\ip{\cdot,\cdot}$, so for forms $\alpha,\beta\in\Omega^p(M)$,
\beq
\ip{\alpha,\beta}=\int_M\alpha\wedge *\beta,\qquad
\|\alpha\|^2=\ip{\alpha,\alpha},
\eeq
where $*:\Omega^p(M)\ra\Omega^{2-p}(M)$ is the Hodge map.
To obtain the usual baby Skyrme model, one chooses $N=S^2\subset\R^3$,
the unit sphere,
with the induced metric and with almost complex structure
$J:T_\phi S^2\ra T_\phi S^2$, $X\mapsto \phi\times X$, so that
$\omega(X,Y)=(\phi\times X)\cdot Y=\phi\cdot(X\times Y)$. In this case,
with respect to any oriented local coordinate system $(x_1,x_2)$ on $M$,
\beq
\phi^*\omega=
\phi\cdot\left(\frac{\cd\phi}{\cd x_1}\times\frac{\cd\phi}{\cd x_2}\right)
\d x_1\wedge\d x_2.
\eeq

We begin by establishing a topological lower energy bound for $E=E_4+E_0$ of
Bogomol'nyi type. The argument has been discovered in particular cases
by several authors \cite{ada1,izq,pie,war1}, and our aim here is to place these 
results
in a general geometric framework.

\begin{prop} For all $\phi:M\ra N$,
$$
E(\phi)\geq\pm\ip{U}\int_M\phi^*\omega
$$
where $\ip{U}$ is the average value of $U$ on $N$, with equality
if and only  if
$$
\phi^*\omega=\pm *U\circ\phi.
$$
\end{prop}

\begin{proof} 
Clearly 
\bea
0&\leq&\frac12\|*\phi^*\omega\mp U\circ\phi\|^2
=E\mp\int_M\phi^*(U\omega).
\eea
By dimensions, $U\omega$ is a closed 2-form and, since $H^2(N)=\R$, 
there exists a
constant $a\in\R$ and $\alpha\in\Omega^1(N)$ such that
\beq
U\omega=a\omega+\d\alpha.
\eeq
Then
\bea
\ip{\omega,U\omega}&=&a\|\omega\|^2+\ip{\omega,\d\alpha}
=a{\rm Vol}(N)+\ip{\dstar\omega,\alpha}
=a{\rm Vol}(N)
\eea
since $\omega$ is coclosed. But
\beq
\ip{\omega,U\omega}=\int_NU\omega=\ip{U}{\rm Vol}(N)
\eeq
where $\ip{U}$ denotes the average value of $U:N\ra\R$.
The result immediately follows.
\end{proof}

We remark that, since $\omega$ is closed, $\int_M\phi^*\omega$ is a
homotopy invariant of $\phi$. In the case of most interest,
$N=S^2$, the bound becomes
\beq
E(\phi)\geq 4\pi\ip{U}|n|
\eeq
where $n\in\Z$ is the degree of $\phi$. 

The Bogomol'nyi equation $\phi^*\omega=*U\circ\phi$ has an interesting
geometric interpretation which we will use frequently in later sections.
Let $N_0=U^{-1}(0)\subset N$, the set of vacua of the model, and $N'=N\less
N_0$, the target space with the vacua removed. We can equip $N'$ with
a deformed area form $\Omega=\omega/U$. Note that this area form blows up
as one approaches $N_0$, the boundary of $N'$. Given a map $\phi:M\ra N$,
denote by $M_\phi$ its critical set, that is
\beq\label{Mphidef}
M_\phi=\{x\in M\: :\: \rank(\d\phi_x)<2\}.
\eeq
At any $x\in M_\phi$, $(\phi^*\omega)_x=0$, since we can always evaluate this
2-form on a basis of vectors one of which is in $\ker\d\phi_x$. Hence,
any solution of the Bogomol'nyi equation maps $M_\phi$ into $N_0$ (sends
critical points to vacua), and on $M'=M\less M_\phi$ satisfies
\beq
\phi^*\Omega=\frac{\phi^*\omega}{U\circ\phi}=*1=\vol_M.
\eeq
That is (as observed for a special case in \cite{pie}):

\begin{remark}\label{nss}
 Bogomol'nyi solutions are area preserving maps from 
$(M',\vol_M)$
to $(N',\Omega)$.
\end{remark}

Note that, as usual, the 
Bogomol'nyi equation is a nonlinear first order PDE for $\phi$.
This is in contrast to the Euler-Lagrange equation for $E$, which is
second order. In analogy with harmonic map theory,
it is convenient to make the following
definition.

\begin{defn} The
{\em tension field} of $\phi:M\ra N$ is
$$
\tau(\phi)=-J\d\phi\sharp\dstar(\phi^*\omega)+(U\grad U)\circ\phi.
$$
Here $\dstar=-*\d*:\Omega^p(M)\ra\Omega^{p-1}(M)$, the coderivative adjoint to
$\d$, and $\sharp$ denotes the metric isomorphism $T^*M\ra TM$ induced
by $g$. 
Note that $\tau(\phi)$ is a section of $\phi^{-1}TN$, the vector bundle over 
$M$ with fibre $T_{\phi(x)}N$ over $x\in M$. We will also consistently
denote the $0$ form $*\phi^*\omega$ by $F_\phi:M\ra \R$, so
$$
\phi^*\omega=F_\phi\vol_M.
$$
\end{defn}

Given a variation $\phi_t$ of $\phi$, with infinitesimal generator
$X=\cd_t\phi_t|_{t=0}\in\Gamma(\phi^{-1}TN)$ a straightforward calculation
\cite{spesve1} shows that
\beq
\left.\frac{d\:}{dt}E(\phi_t)\right|_{t=0}=\ip{X,\tau(\phi)}=\int_M h(X,\tau(\phi))\vol_M.
\eeq
Hence, the Euler-Lagrange equation is
\beq\label{eom}
\tau(\phi)=0.
\eeq
Any solution of the Bogomol'nyi equation
\beq
\label{bog}
F_\phi=\pm U\circ\phi
\eeq
minimizes energy in its
homotopy class, so must satisfy the field equation (\ref{eom}) by the
fundamental lemma of the calculus of variations. It is reassuring to
verify this fact directly. The key observation is contained in the
following lemma.

\begin{lemma}
\label{key}
Let $\phi:M\ra N$ and $X$ be a vector field on $M$. Then
$$
h(\d\phi X,\tau(\phi))=-\frac12\d(F_\phi^2-(U\circ\phi)^2)X.
$$
\end{lemma}

\begin{proof}   One sees that
\beq\label{q}
\sharp\dstar\phi^*\omega=-\sharp *\d *\phi^*\omega=
-\sharp *\d F_\phi
=-J_M\grad F_\phi
\eeq
where $J_M$ is the almost complex structure induced by the orientation
on $M$. Hence
\bea
h(\d\phi X,\tau(\phi))&=&
h(\d\phi X,J\d\phi J_M\grad F_\phi+(U\grad U)\circ\phi)\nonumber \\
&=&
-\phi^*\omega(X,J_M\grad F_\phi)+(U\d U)(\d\phi X)\nonumber\\
&=&
-F_\phi g(X,\grad F_\phi)+(U\circ\phi)g(X,\grad(U\circ\phi))\nonumber\\
&=&
-\frac12g(X,\grad(F_\phi^2-(U\circ\phi)^2))=-\frac12\d(F_\phi^2-(U\circ\phi)^2)X.
\eea
\end{proof}

{We remark that this Lemma remains true under the weaker assumption that
$V=\frac12 U^2$ is $C^1$ (rather than $U$ itself). One replaces
$U\grad U$ and $U\d U$ by $\grad V$ and $\d V$ throughout the proof.}

\begin{prop}\label{prop1} Let $\phi:M\ra N$ satisfy (either of the)
Bogomol'nyi equation(s), $\phi^*\omega=\pm * U\circ\phi$ everywhere.
Then $\phi$ satisfies the field equation $\tau(\phi)=0$.
\end{prop}

\begin{proof} By assumption $F_\phi^2-(U\circ\phi)^2$ is constant on $M$, so
by Lemma \ref{key} we have that $h(\d\phi_x X,\tau(\phi)(x))=0$
for all $x\in M$ and all $X\in T_xM$. It follows that $\tau(\phi)(x)=0$
at all regular points of $\phi$, since $\d\phi_x(T_xM)=T_xN$ at such $x$.
It remains to show that $\phi$ satisfies (\ref{eom}) on its critical
set. So, let $x$ be a critical point of
$\phi$ (meaning $\rank\d\phi_x<2$). 
Then
$\phi^*\omega_x=0$ so $F_\phi(x)=0$, and hence $\phi(x)\in U^{-1}(0)$. 
But $U\geq 0$, so $\phi(x)$ is a minimum of $U$, and hence
$(\grad U)(\phi(x))=0=\d U_{\phi(x)}$ also. 
Hence $(\grad F_\phi)(x)=\sharp(\d F_\phi)_x=\pm\sharp\d U_{\phi(x)}\d\phi_x=0$,
and one sees from equation (\ref{q}) that $\dstar\phi^*\omega=0$ 
at $x$. Hence $\phi$
satisfies (\ref{eom}) at $x$.
\end{proof}

So solutions of the Bogomol'nyi equation automatically satisfy the
field equation, as usual. In a general field theory of
Bogomol'nyi type, there is no reason why
solutions of the field equation should necessarily satisfy the
Bogomol'nyi equation. Remarkably, we will show that, on $M=\R^2$, all
solutions of the field equation satisfy one or other of the
the Bogomol'nyi equations at each point. 

\begin{prop}\label{prop2} Let $\phi:\R^2\ra N$ satisfy the field equation
(\ref{eom}) and boundary condition (\ref{bc}). Then
$$
F_\phi^2=(U\circ\phi)^2
$$
everywhere
\end{prop}

\begin{proof}
Since $\tau(\phi)=0$, we see from Lemma \ref{key} that
  $F_\phi^2-(U\circ\phi)^2$ is constant
on $\R^2$. But $\phi(x)\ra\phi_0\in U^{-1}(0)$ as $|x|\ra \infty$ by
(\ref{bc}), so $F_\phi(x)\ra 0$ and $U(\phi(x))\ra 0$ as $|x|\ra \infty$.
Hence, this constant is $0$.
\end{proof}

Since its proof uses only Lemma \ref{key}, Proposition \ref{prop2} extends 
immediately to the weaker case that $V=\frac12 U^2$ is $C^1$.
By contrast, the proof of Proposition \ref{prop1} 
makes essential use of the differentiability
of $U$, so does not extend to this weaker case.
By the {\em support} of a map $\phi:M\ra N$ we mean the closure of
$\phi^{-1}(N')$, that is
\beq
\supp\phi={\rm closure}(\{x\in M\: :\: \phi(x)\notin N_0\}).
\eeq
It follows immediately from Proposition \ref{prop2} that, for a solution 
$\phi$ on $M=\R^2$, $\supp\phi={\rm closure}(M\less M_\phi)$.

As we
will see later, it is possible, for suitable $U$,
 to construct solutions to (\ref{eom}) by
gluing together maps with  $\phi^*\omega=*U\circ\phi$ and
$\phi^*\omega=-*U\circ\phi$ in different regions of $M$. So it does not
follow from Proposition \ref{prop2}
that all solutions of the theory are global energy minimizers. In particular,
the vacuum sector can contain infinitely many static solutions, of 
arbitrarily
high energy. Remarkably, we will see that these ``lump-antilump''
superpositions are actually local {\em minima} of $E$, not saddle
points. A key property which we will exploit in the construction
of these exotic multilumps is the invariance of the model under area
preserving diffeomorphisms of $(M,g)$. Once again, this
property has been observed previously in specific cases by many authors
\cite{gis,pie}.

\begin{prop}\label{prop3}
 Let $\phi:M\ra N$ and $\aa:M'\ra M$ be an area preserving
diffeomorphism. Then $E(\phi\circ\aa)=E(\phi)$.
\end{prop}

\begin{proof} By assumption, $\aa^*\vol_M=\vol_{M'}$. Let $\psi=\phi\circ\aa$.
Then
\bea
\psi^*\omega=\aa^*(\phi^*\omega)=\aa^*(F_\phi\vol_M)=
(F_\phi\circ\aa)\vol_{M'}.
\eea
Hence
\bea
E_4(\psi)=\frac12\int_{M'}(F_\phi\circ\aa)^2\vol_{M'}
=\frac12\int_{M'}\aa^*(F_\phi^2\vol_M)
=\frac12\int_M F_\phi^2\vol_M=E_4(\phi)
\eea
since $\aa:M'\ra M$ is a diffeomorphism. Similarly,
\bea
E_0(\psi)&=&\frac12\int_{M'}(U\circ\phi\circ\aa)^2\vol_{M'}
=\frac12\int_{M'}\aa^*((U\circ\phi)^2\vol_M)\nonumber \\
&=&
\frac12\int_{M}(U\circ\phi)^2\vol_M=E_0(\phi).
\eea
\end{proof}

It follows immediately that $\tau(\phi\circ\aa)=\tau(\phi)\circ\aa$ 
so $\phi\circ\aa$ satisfies the field equation (\ref{eom}) if and only if
$\phi$ does. Since $F_{\phi\circ\aa}=F_\phi\circ\aa$, we also verify immediately
that $\phi\circ\aa$ satisfies the Bogomol'nyi equation (\ref{bog}) if and
only if $\phi$ does.

\section{Semi-compactons}
\news

In this section we restrict attention to the case $M=\R^2$, $N=S^2$ and
\beq\label{U}
U(\phi)=1+\phi_3,
\eeq
 though the constructions below clearly generalize to any
$U$ which is $C^1$, non-negative and has a single non-degenerate zero.
It is straightforward \cite{ada1}
to find a degree $1$ solution of the Bogomol'nyi
equation (\ref{bog}) within the hedgehog ansatz
\beq\label{hhog}
\phi(r,\theta)=(\sqrt{1-z(r)^2}\cos\theta,\sqrt{1-z(r)^2}\sin\theta,z(r))
\eeq
where $(r,\theta)$ are polar coordinates on $\R^2$, $z:[0,\infty)\ra\R$,
$z(0)=1$ and $z(\infty)=-1$. In terms of cylindrical
coordinates $Z=\phi_3$ and $\Theta={\rm arg}(\phi_1+i\phi_2)$, the K\"ahler
form on $S^2$ is $\omega=\d\Theta\wedge\d Z$, and the ansatz (\ref{hhog})
is $\Theta=\theta$, $Z=z(r)$. Hence, the Bogomol'nyi
equation (\ref{bog}) becomes
\beq
-\frac{z'}{r}=1+z
\eeq
whose solution, with the required boundary data, is
\beq\label{lump}
z(r)=-1+2e^{-r^2/2}.
\eeq
Note that this solution has faster than exponential decay, and is
smooth everywhere, including at the origin. To check this, define the
(globally) analytic function
\beq
q(s)=\sum_{n=0}^\infty \frac{(-1)^n}{2^{n+1}(n+1)!}s^{n}
\eeq
and note that
\beq
\sqrt{1-z(r)^2}=2r(1-r^2q(r^2))^2\sqrt{q(r^2)}.
\eeq
Since $q(0)\neq 0$, $\sqrt{q(s)}$ is analytic
in a neighbourhood of
$0$, so
\beq
\sqrt{1-z(r)^2}=rQ(x^2+y^2)
\eeq
where $Q$ is analytic on a neighbourhood of $0$. It follows that
\beq
\phi(x,y)=(xQ(x^2+y^2),yQ(x^2+y^2),-1+2e^{-(x^2+y^2)/2})
\eeq
is smooth at $(0,0)$.
Clearly, $\ip{U}=1$, so this unit lump solution has energy $E=4\pi$. 

One can seek degree $n\geq 2$ solutions within the ansatz (\ref{hhog}) by
replacing $\theta$ with $n\theta$, as in \cite{ada1},
\beq
\phi(r,\theta)=(\sqrt{1-z_n(r)^2}\cos n\theta,\sqrt{1-z_n(r)^2}\sin n\theta,
z_n(r))
\eeq
The profile function is then
$z_n(r)=-1+2\exp(-r^2/2n)$. But such fields are not even once
differentiable at the origin, so are not genuine solutions of
the Bogomol'nyi (or field) equation in the sense that we demand. 
The problem is that $\phi$ has a conical singularity at $(0,0)$. To see
this, let
\beq\label{stereo}
W=\frac{\phi_1+i\phi_2}{1+\phi_3}
\eeq
be the image of $\phi$ under stereographic projection from $(0,0,-1)$.
Note that $W$ is a good complex coordinate on a neighbourhood of
$\phi(0,0)=(0,0,1)$. Then, for this radially symmetric $n$-lump,
\beq
W=\left(\frac{r}{\sqrt{2n}}+O(r^3)\right)e^{in\theta}
=\left(\sqrt{\frac{x^2+y^2}{2n}}+O(r^3)\right)\left(\frac{x+iy}{\sqrt{x^2+y^2}}
\right)^n.
\eeq
Hence
\beq
W_x(0,y)=\sqrt{\frac{n}{2}}\left(\frac{iy}{|y|}\right)^{n-1}+O(y^2)
\eeq
which has a step discontinuity at $y=0$. There is a similar problem
with the radially symmetric degree $n$ solutions obtained in
\cite{pie}.
We will see below that genuine (at least twice
differentiable) solutions of the Bogomol'nyi equation do exist for each 
$n\geq1$ but constructing them requires some ingenuity. 

A geometric
insight into the difficulty one faces can be obtained from Remark \ref{nss}.
In this case, the vacuum manifold is $N_0=\{(0,0,-1)$, so 
$N'=N\less N_0$ is a punctured sphere or, equivalently, an open disk.
The deformed area form on $N'$ is, in cylindrical coordinates,
\beq
\Omega=\frac{\d\Theta\wedge\d Z}{1+Z}
\eeq
which gives $N'$ infinite total area. In fact, $(N',\Omega)$ 
can be visualized as a
``cigar shaped'' surface of revolution,
 with a single infinite cylindrical end replacing the
missing point $N_0$, see figure \ref{cigar}(a). 
This comes from identifying $\Omega$ with the area form on the punctured
sphere $N'$ associated with the metric
\beq
h'=\frac{2\d W\d\bar{W}}{1+|W|^2},
\eeq
where $W$ is the stereographic coordinate defined in (\ref{stereo}).
The degree $1$ energy
minimizer constructed above can now be seen as an area-preserving
diffeomorphism from $\R^2$ to $(N',\Omega)$. The difficulty in 
constructing higher degree solutions is that any map of degree exceeding $1$
must have critical points. Any such critical point must get mapped to $N_0$, 
the end at infinity, and it is hard to 
arrange this while maintaining the area-preserving property of $\phi$ away
from its critical points. Certainly $\phi$ cannot have any isolated
critical points (as a generic map between 2-manifolds does),
since we have the following proposition.

\begin{figure}
\begin{center}
\begin{tabular}{cccc}
\includegraphics[scale=0.45]{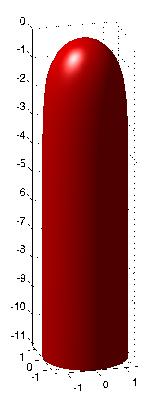}&
\includegraphics[scale=0.45]{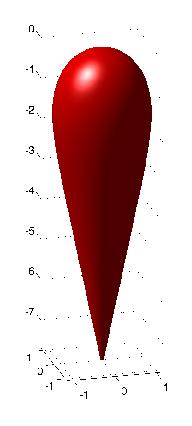}&
\includegraphics[scale=0.45]{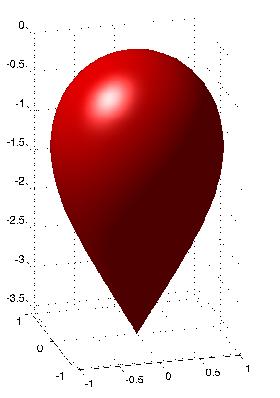}&
\includegraphics[scale=0.45]{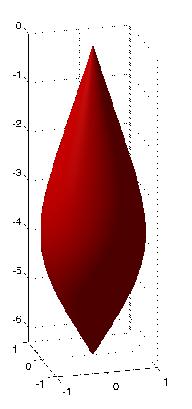}\\
(a)&(b)&(c)&(d)
\end{tabular}
\end{center}
\caption{The deformed target spaces $(N',\Omega)$ embedded as surfaces of
revolution, in the cases (a) $U=1+\phi_3$, (b) $U=(1+\phi_3)^{0.8}$,
(c) $U=(1+\phi_3)^{0.5}$ and (d) $U=(1+\phi_3)^{0.5}(1-\phi_3)^{0.7}$.}
\label{cigar}
\end{figure}

\begin{prop} Let $\phi:\R^2\ra S^2$ be a solution of the model
with $U(\phi)=1+\phi_3$ satisfying boundary condition (\ref{bc}). Then 
every connected component of the critical set of
$\phi$ is unbounded.
\end{prop}

\begin{proof} By Proposition \ref{prop2}, $F_\phi^2=U^2$ everywhere and so
$\phi$ maps the critical set
$M_\phi$ into $N_0=U^{-1}(0)$, the vacuum manifold. 
Assume, towards a contradiction, that $M_\phi$ has a bounded
connected component $M_1$. Then for $\eps>0$ sufficiently small
the closed 1-manifold $\phi_3^{-1}(-1+\eps)$ has a connected component
$\Gamma\cong S^1$ whose interior contains $M_1$. Let $S$ be the interior
of $\Gamma$ with $M_1$ removed, and consider the restriction of $\phi$ to $S$.
By Remark \ref{nss} this is an area-preserving surjective map from
$S$ to $N_\eps=\{\phi\: :\: -1<\phi_3<-1+\eps\}$ with respect to $\Omega$.
But $S$, being a bounded subset of $\R^2$, has finite area while 
$(N_\eps,\Omega)$ has inifnite area, a contradiction.
\end{proof}

Nonetheless, this model does have solutions in every homotopy
class.
We construct them as follows. Let $\aa:(0,\infty)\times\R\ra\R^2$ be the
diffeomorphism
\beq
\aa:(x,y)\mapsto (\log x,xy).
\eeq
This map is area-preserving (with respect to the Euclidean metric on 
both spaces). Let $\psi:\R^2\ra S^2$ denote the unit lump
solution constructed above, equations (\ref{hhog}), (\ref{lump}). Then
as remarked after Proposition \ref{prop3}, $\psi\circ\aa$ satisfies the
Bogomol'nyi equation on the half-space $(0,\infty)\times\R$. Clearly,
$\lim_{x\ra0^+}\psi(\aa(x,y))=(0,0,-1)$ for all $y$. Hence the map
\beq\label{half}
\phi:\R^2\ra S^2,\qquad
\phi(x,y)=\left\{\begin{array}{cc}\psi(\aa(x,y))&x>0\\
(0,0,-1)&x\leq0\end{array}\right.
\eeq
is continuous and satisfies the Bogomol'nyi equation
away from the line $x=0$. We claim that this is a genuine degree $1$
solution of the Bogomol'nyi equation, and hence, the field equation.
This amounts to the claim that $\phi$ is twice continuously differentiable
everywhere.

\begin{prop} The mapping $\phi:\R^2\ra S^2$ defined in equation
(\ref{half}) is $C^2$.
\end{prop}

\begin{proof}
 Clearly, $\phi$ is smooth away from the line $x=0$, and
all its derivatives vanish identically for $x<0$. So it suffices to
show that
\beq
|\phi_x|,|\phi_y|,|\phi_{xx}|,|\phi_{xy}|,|\phi_{yy}|
\eeq
all vanish in the limit $x\ra 0^+$, for all $y$. 
For $x>0$ we have that $\phi(x,y)=\psi(X,Y)$ where $X=\log x$, $Y=xy$.
Straightforward
estimates using the explicit formulae (\ref{hhog}),(\ref{lump})
yield that there exist constants $C,R_*>0$ such that
for all $R=\sqrt{X^2+Y^2}\geq R_*$,
\bea
|\psi_X|,|\psi_Y|&\leq&CRe^{-R^2/4}\nonumber \\
|\psi_{XX}|,|\psi_{XY}|,|\psi_{YY}|&\leq&CR^2e^{-R^2/4}.\label{est1}
\eea
Similarly, there exists constant $X_*<0$ such that for all $0<x<e^{X_*}$,
\bea
|X_x|,|X_y|,|Y_x|,|Y_y|&\leq&|y|+e^{-X}\nonumber\\
|X_{xx}|,|X_{xy}|,|X_{yy}|,|Y_{xx}|,|Y_{xy}|,|Y_{yy}|&\leq&e^{-2X}.\label{est2}
\eea
Hence, by the chain rule, for all $0<x<x_*=\min\{e^{X_*},e^{-R_*}\}$ and
all
$y$
\bea
|\phi_x|&\leq& CRe^{-R^2/4}(|y|+e^{-X})
\leq C(|X|+ye^{X})e^{-X^2/4}(y+e^{-X})\ra 0
\eea
as $x\ra 0^+$, since then $X\ra -\infty$. Hence $\lim_{x\ra 0^+}|\phi_x(x,y)|
=0$ for all $y$. The same argument deals with $\phi_y$.

Turning to the second derivatives, we see from the chain rule and estimates
(\ref{est1}), (\ref{est2}) that for all $0<x<x_*$ and
all
$y$
\bea
|\phi_{xx}|&\leq&C\left\{Re^{-R^2/4}e^{-2X}+R^2e^{-R^2/4}(y^2+e^{-2X})\right\}\nonumber
\\
&\leq& C(|X|+y+X^2+y^2)(y^2+e^{-2X})e^{-X^2/4}\ra 0
\eea
as $x\ra 0^+$, since then $X\ra -\infty$. Hence $\lim_{x\ra 0^+}|\phi_{xx}(x,y)|
=0$ for all $y$. The same argument deals with $\phi_{xy},\phi_{yy}$.
\end{proof}

It seems likely that the mapping $\phi$ defined in (\ref{half}) is
actually smooth everywhere, but we have not proved this. Let us henceforth
denote this degree 1 $C^2$ map, which satisfies the Bogomol'nyi equation 
everywhere, $\phi_+$. Note that
$E(\phi_+)=4\pi$, the topological minimum value in its homotopy class.
Since it takes exactly the vacuum value on the
left half-plane, one could call this solution a {\em semi-compacton}.
However, by exploiting the invariance of $E$ under area-preserving
diffeomorphisms further, we can construct degree 1 energy minimizers
with more tightly localized support.

 Consider the map
\beq
\aa':(0,\infty)\times(-\frac\pi2,\frac\pi2)\ra (0,\infty)\times\R,\qquad
\aa'(x,y)=
(x\cos^2y,\tan y).
\eeq
Clearly $\aa'$ is an area-preserving diffeomorphism. For any $\eps>0$, denote
by $\phi_+^\eps$ the $x$-translate of $\phi_+$ by $\eps$, that is,
\beq
\phi_+^\eps(x,y)=\phi_+(x-\eps,y).
\eeq
The support of $\phi_+^\eps$ is the half plane $x\geq\eps$. 
Denote by $S$ the infinite half-strip
$S=(0,\infty)\times(-\frac\pi2,\frac\pi2)$, and consider the
mapping
\beq
\phi_\sqsubset:\R^2\ra S^2,\qquad
\phi_\sqsubset(x,y)=\left\{\begin{array}{cc}
\phi_+^\eps(\aa'(x,y))&(x,y)\in S\\
(0,0,-1)&(x,y)\notin S.\end{array}\right.
\eeq
By construction, this is continuous everywhere and $C^2$ on the complement
of $\cd S$, the boundary of the strip $S$. It also satisfies the Bogomol'nyi
equation on $\R^2\less\cd S$. By construction, its support is a subset of the
closure of $S$. In fact, 
\beq
\supp\phi_\sqsubset=\{(x,y)\: :\: x\geq \eps/\cos^2y\}\subset S.
\eeq
Hence $\phi_\sqsubset$ 
is constant on a neighbourhood of $\cd S$, and so is trivially
$C^2$ on $\cd S$. Hence $\phi_\sqsubset$ is $C^2$ everywhere. By construction,
$E(\phi_\sqsubset)=4\pi$, that is, $\phi_\sqsubset$ is a degree 1 energy
minimizer, which we call a {\em semi-compacton}. It has a single energy
density maximum located at the point $(1+\eps,0)$. The energy density
along the line $y=0$ is
\beq
\ee_\sqsubset(x,0)=\left\{\begin{array}{cc}4(x-\eps)^{-\log(x-\eps)}
& x>\eps\\
0& x\leq\eps.\end{array}\right.
\eeq
So $\phi_\sqsubset$ has an energy tail which decays along the strip $S$
like $x^{-\log x}$, faster than any power, but slower than exponential.
The energy density of $\phi_\sqsubset$ (for $\eps$ very small) is
plotted in figure \ref{tail}. 

\begin{figure}
\begin{center}
\includegraphics[scale=0.5]{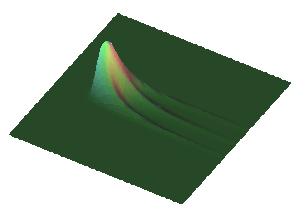}
\end{center}
\caption{The energy density of a semi-compacton.}
\label{tail}
\end{figure}

 By precomposing $\phi_\sqsubset$ with an area preserving 
diffeomorphism
\beq
\aa:\R^2\ra\R^2, \qquad
\aa(x,y)=(\alpha x, \alpha^{-1}y)
\eeq
where $\alpha>0$, we can construct semi-compactons with support in an
arbitrarily thin, half infinite strip. Similarly, the strip can be
deformed to follow any non-self-intersecting half infinite curve which 
escapes to infinity. The mapping
$\bar\phi_{\sqsubset}(x,y)=\phi_\sqsubset(x,-y)$ is an anti-semi-compacton, of 
degree
$-1$. By gluing together (anti-)semi-compactons with disjoint support,
one obtains $C^2$ energy minimizers in every homotopy class.
Gluing together $n_+>0$ semi-compactons and $n_->0$ anti-semicompactons
yields degree $n=n_+-n_-$ fields which, by Proposition \ref{prop1},
are $C^2$ solutions of the field equation, but have energy
$4(n_++n_-)\pi>4|n|\pi$. So each homotopy class contains critical
points of $E$ of arbitrarily high energy. Even more surprising,
these critical points are {\em not} saddle points of $E$ but are, in
a certain sense, linearly stable. 

To see this, one must construct the Hessian operator for
the functional $E(\phi)$ based at a critical point $\phi$. We recall that
this is defined as follows. Let $\phi_{s,t}$ be a two-parameter
variation of a critical point $\phi:M\ra N$ of $E$, and
let $X=\cd_s\phi_{s,t}|_{s=t=0}$, 
$Y=\cd_t\phi_{s,t}|_{s=t=0}\in\Gamma(\phi^{-1}TN)$
be the associated inifnitesimal variations. Then the Hessian of $E$
at $\phi$ is the symmetric bilinear form
\beq
{\rm Hess}(X,Y)=\left.\frac{\cd^2 \:}{\cd s\cd t}E(\phi_{s,t})\right|_{s=t=0}
\eeq
on $\Gamma(\phi^{-1}TN)$. The associated Hessian operator is the 
self-adjoint
linear differential operator $\hh:\Gamma(\phi^{-1}TN)\ra\Gamma(\phi^{-1}TN)$
defined such that
\beq
{\rm Hess}(X,Y)=\int_M h(X,\hh Y)\vol_M=\ip{X,\hh Y}.
\eeq
One uses the spectrum of $\hh$ to classify the critical point $\phi$.
In particular, if $\hh$ has both negative and positive eigenvalues,
$\phi$ is a saddle point. If the quadratic form ${\rm Hess}(X,X)$ is 
non-negative, one says that $\phi$ is linearly stable (although
$\phi$ may actually be dynamically unstable; e.g.\ $0$ is a
linearly stable critical point of $f(x)=-x^4$). 

For the energy under consideration here, one finds that \cite{spesve1}
\beq
\hh Y=-J(\nabla^\phi_{Z_\phi}Y+\d\phi(\sharp\dstar\d\phi^*\iota_Y\omega))
+(\nabla^N_Y\grad V)\circ\phi
\eeq
where $Z_\phi=\sharp\dstar\phi^*\omega\in\Gamma(TM)$,  $\nabla^N$ is the
Levi-Civita connexion on $TN$, $\nabla^\phi$
is its pullback to $\phi^{-1}TN$, $V=\frac12 U^2$ and $\iota$ denotes interior
product ($\iota_A\omega=\omega(A,\cdot)$). The exact details of this formula
are not important. We will need only the following Lemma.

\begin{lemma}\label{freaky}
Let $\phi:M\ra N$ be a critical point of $E$, $\hh$ be its Hessian
operator and $x\in (M\less\supp\phi)$. Then, for all $Y\in\Gamma(\phi^{-1}TN)$,
$$
(\hh Y)(x)=0.
$$
That is, the Hessian operator vanishes identically off the support of $\phi$.
\end{lemma}

\begin{proof} The complement of $\supp\phi$ is open by definition, so
$\phi$ is constant on a neighbourhood of $x$. It follows that
$Z_\phi=0$ and $\d\phi=0$ on a neighbourhood of $x$, so the first two terms in
$\hh Y$ vanish at $x$ for all $Y$. Consider now the zeroth order piece
\beq
\hh_0Y=(\nabla^N_Y\grad V)\circ\phi.
\eeq
Since $V=\frac12U^2$,
\beq
\hh_0Y=(\nabla^N_Y(U\grad U))\circ\phi
=(h(Y,\grad U)\grad U)\circ\phi+(U\nabla^N_Y\grad U)\circ\phi.
\eeq
But $U(\phi(x))=0$ (since $x\notin\supp\phi$) and, since $U$ is assumed 
non-negative, $\phi(x)$ is a minimum of $U$, and hence $(\grad U)(\phi(x))=0$
also. Hence, for all $Y$, $(\hh_0 Y)(x)=0$
\end{proof}

\begin{prop}
Let $\phi:M\ra N$ be any solution of (\ref{eom}) constructed by
superposing (anti-)\newline semi-compactons $\phi_i$ with support in disjoint
strips $S_i$, $i=1,2,\ldots,m$.  Then for all $Y\in\Gamma^{-1}(TN)$,
$$
{\rm Hess}(Y,Y)\geq 0.
$$
\end{prop}

\begin{proof} By Lemma \ref{freaky},
\beq
{\rm Hess}(Y,Y)=\sum_{i=1}^m\int_{S_i}h(Y,\hh_i Y)\vol_M
\eeq
where $\hh_i$ denotes the Hessian operator associated to the 
(anti-)semi-compacton $\phi_i$. Each term in this sum is non-negative for
all $Y$. For if not, then, by Lemma \ref{freaky},
 there exists $i$ and a section $Y$ such that
\beq
\ip{Y,\hh_iY}=\int_Mh(Y,\hh_i Y)\vol_M=\int_{S_i}h(Y,\hh_i Y)\vol_M<0,
\eeq
which contradicts the fact that $\phi_i$ minimizes $E$ in its homotopy 
class.
\end{proof}

 Physically, the point is that semi-compactons exert no forces
on one another, so $\phi_{\sqsubset}\bar\phi_\sqsubset$ superpositions
are (marginally)
stable, by stability of their constituent parts.

So this model supports degree $n$ (marginally) stable multi-semi-compactons of
energy $4(|n|+2k)\pi$ for all $n\in \Z$ and $k\in\Z_{\geq 0}$. All these
solutions have (multiple) tails  escaping to infinity, along which the
energy
density decays like $x^{-\log x}$. 

\section{Compactons revisited}
\news

When will an extreme baby Skyrme model support genuine compactons? 
The geometric picture outlined above immediately gives
a necessary condition on $U$, namely that $N'=N\less N_0$ (the target space
with its vacua removed) should have finite volume with respect to the
deformed area form $\Omega=\omega/U$. 

\begin{prop} Let $\phi:\R^2\ra N$ be a surjective solution of $\tau(\phi)=0$
of compact support. Then $(N',\Omega)$ has finite volume.
\end{prop}

\begin{proof}
By Proposition \ref{prop2}, $F_\phi^2=(U\circ\phi)^2$ everywhere. Since 
$F_\phi(x)=0$ if and only if $x\in M_\phi$, $\phi$ defines an area
preserving map from each connected component of 
$M\less M_\phi$ into $(N',\Omega)$. The union of the ranges of each such map
is all $N'$ (since $\phi$ is surjective), and hence the area of $N'$
cannot exceed the area of $M\less M_\phi\subset\supp\phi$.
\end{proof}

Conversely, if $(N',\Omega)$ has finite area $A$, let $M'$ be any subset
of $M=\R^2$ of area $A$ which is diffeomorphic to $N'$. For example, if 
$N_0$ consists of $p$ vacua, one could take $M'$ to be an open disk of
area $A+p\eps$ with $p$ small disjoint closed disks of area $\eps$ removed.
Construct an area-preserving diffeomorphism $\psi:M'\less N'$, using the
method of Moser, for example \cite{mos}, and extend $\psi$ to the whole
of $M$ by a piecewise constant map on $M\less M'$. This map $\phi$
certainly has
compact support, and satisfies the field equation except, perhaps, on the
boundary of $M'$. Hence $\phi$ is a genuine solution if and only if it
is $C^2$.

For example, consider the model with
\beq
U(\phi)=(1+\phi_3)^\alpha
\eeq
where $\frac12\leq\alpha<1$. Here $N'=S^2\less\{(0,0,-1)\}$, diffeomorphic to an
open disk, but, unlike the
case $\alpha=1$ considered in the previous section, $N'$ now
has finite volume,
\beq
{\rm Vol}(N')=\int_{N'}\frac{\omega}{U}=
\int_{N'}\frac{\d\Theta\wedge\d Z}{(1+Z)^\alpha}
=\frac{2^{2-\alpha}}{1-\alpha}\pi.
\eeq
It can be visualized as a baloon shaped surface of revolution, with a conical
singularity at the missing vacuum point, see figure \ref{cigar}(b),(c).
An obvious choice for the open set $M'$ is the disk of radius $R=2^{1-\alpha/2}
(1-\alpha)^{-\frac12}$. There is an area-preserving diffeomorphism $M'\ra
N'$ within the radial ansatz (\ref{hhog}),
\beq
z(r)=\left[2^{1-\alpha}-\frac12(1-\alpha)r^2\right]^{\frac{1}{1-\alpha}}-1,
\eeq
which, when extended by $(0,0,-1)$ outside the disk $M'$ gives a $C^2$
map $\R^2\ra S^2$ of degree $1$ solving the field equation everywhere.
This (up to reparametrization) is the compacton reported by 
Adam et al \cite{ada1}. Note that one can obtain its key qualitative features
without solving any equations, e.g.\ it occupies area 
$\frac{2^{2-\alpha}}{1-\alpha}\pi$ and has total energy
\beq
E=
4\pi\ip{U}=\frac{4\pi}{{\rm Vol}(S^2)}\int_{S^2}(1+Z)^\alpha \d\Theta\wedge
\d Z
=\frac{2^{\alpha+2}}{\alpha+1}\pi.
\eeq

Another interesting choice is
\beq\label{funkyU}
U(\phi)=(1+\phi_3)^\alpha(1-\phi_3)^\beta
\eeq
where $\alpha,\beta\in[\frac12,1)$. Now $N'$ is diffeomorphic to a
cylinder
and has finite total area $A(\alpha,\beta)$,
a complicated function of $\alpha, \beta$ involving hypergeometric functions.
An embedding of $N'$ as a surface of revolution
in the case $\alpha=0.5$, $\beta=0.7$  is depicted in figure \ref{cigar}(d).
One can take $M'$ to be any annulus
of total area $A$, 
\beq
M'=\{(x,y)\: :\:  R_1^2<x^2+y^2<R_2^2\}\qquad
\mbox{
where}\quad \pi(R_2^2-R_1^2)=A,
\eeq
 and 
construct an area-preserving diffeomorphism
$\psi:M'\ra N'$ within the ansatz (\ref{hhog}), then extend this
by $(0,0,-1)$ for $r\leq R_1$, and $(0,0,1)$ for $r\geq R_2$.
It is straightforward to check that this field is $C^2$, and hence 
defines a ringlike compacton. By choosing $R_2$ sufficiently large (and
$R_1$ close to $R_2$) this ring can be arbitrarily big. Hence one can construct
$n$-compactons, with $n$ rings nested inside one another, as well as
the more obvious multi-ring solutions. Adam et al consider only the
degenerate case that the annulus is a punctured disk \cite{ada1}, so we
shall go through this construction in more detail. 

The deformed area form (in cylindrical coordinates) is
$\Omega=(1+Z)^{-\alpha}(1-Z)^{-\beta}\d\Theta\wedge\d Z$, so a field within the
ansatz (\ref{hhog}) satisfies the Bogomol'nyi equation if and only if
\beq\label{qw}
\frac{z'}{(1+z)^\alpha(1-z)^\beta}=-r.
\eeq
Define the function
\beq
Q:[-1,1]\ra [0,A/(2\pi)],\qquad
Q(Z)=\int_{-1}^Z\frac{dt}{(1+t)^\alpha(1-t)^\beta}.
\eeq
Then $Q^{-1}$ is an increasing, surjective
 $C^2$ map $[0,A/(2\pi)]\ra [-1,1]$, and
\beq
z(r)=Q^{-1}(C-\frac{r^2}{2})
\eeq
solves (\ref{qw}) for any constant $C>0$. We require $\phi(0,0)=(0,0,1)$,
so insist that $C>A/(2\pi)$. Set $R_1=\sqrt{2C}$ and define $R_2>R_1$
such that $\pi(R_2^2-R_1^2)=A$. Then the
extended profile function is
\beq
z(r)=\left\{\begin{array}{cc}
1 & 0\leq r\leq R_1\\
Q^{-1}(\frac12(R_1^2-r^2)) & R_1<r<R_2\\
-1 & r\geq R_2\end{array}\right.
\eeq
Note that the associated field has support in an annulus of total area $A$,
as expected. Note also that it is constant in a neighbourhood of the
(polar) coordinate singularity at $r=0$, so to check that $\phi$ is $C^2$,
it suffices to check that $z(r)$ is $C^2$. This is clear, except at the 
points $r=R_1$, where $z=-1$ and $r=R_2$, where $z=1$.  By the Bogomol'nyi
equation,
\beq\label{er}
z'(r)=-r(1+z(r))^\alpha(1-z(r))^\beta
\eeq
on $(R_1,R_2)$, whence $\lim_{r\ra R_1^+}z'(r)=\lim_{r\ra R_2^-}z'(r)=0$.
Hence $z(r)$ is $C^1$. Differentiating (\ref{er}),
\bea
z''(r)&=&-(1+z(r))^\alpha(1-z(r))^\beta\nonumber\\
&&-r^2\{\alpha(1+z(r))^{2\alpha-1}(1-z(r))^{2\beta}
-\beta(1+z(r))^{2\alpha}(1-z(r))^{2\beta-1}\}
\eea
on $(R_1,R_2)$, whence $\lim_{r\ra R_1^+}z''(r)=\lim_{r\ra R_2^-}z''(r)=0$ also
(note $\alpha,\beta\geq \frac12$). Hence $z(r)$ is $C^2$.

One can precompose this map with an arbitrary area-preserving
diffeomorphism $\aa:\R^2\ra\R^2$ to obtained deformed ring-like compactons.
Choosing $R_2$ very close to $R_1$, then deforming, produces
closed  string-like compactons.
In fact, one can precompose it with a degree $n$ area-preserving
covering map
\beq
\aa:\R^2\less\{(0,0)\}\ra R^2\less \{(0,0)\},\qquad
\aa:(r,\theta)\mapsto (n^{-\frac12}r,n\theta)
\eeq
to obtain a degree $n$ annular compacton. Unlike the single vacuum case,
this is still $C^2$ (even at the origin) because the degree $1$ compacton
is constant on a neighbourhood of the origin.

Finally, consider the case of potential (\ref{funkyU}) in the case
$\beta=1$. This supports a $C^2$ degree 1 energy minimizer which decays
like $\exp(-r^2/2)$ to $(0,0,1)$ as $r\ra\infty$, and is exactly
$(0,0,-1)$ on any closed disk centred on the origin. By precomposing
this with appropriate area preserving maps, as in the previous section, we
can produce a semi-compacton localized in a semi-infinite strip, with
a $x^{-\log x}$ tail, but with a hole (of any finite area) in the middle of
the lump, where it has exactly zero energy. Clearly, by introducing
more vacua, one can dream up models with even more bizarre energy
minmizers.

\section{Concluding remarks}
\news

We have shown that the extreme baby-Skyrme model with energy 
\beq
E=\frac12\int_{\R^2}\{[\phi\cdot(\phi_x\times\phi_y)]^2+U(\phi)^2\}dx\, dy,
\qquad U(\phi)=1+\phi_3
\eeq
supports, in every homotopy class, semi-compacton solutions of
quantized energy $E=4\pi(|n|+2k)$ where $n\in\Z$ is the degree of $\phi$
and $k$ is a non-negative integer. These solutions are (at least) twice
continuously differentiable everywhere, and consist of $|n|+2k$ (anti-)lumps,
each localized in a semi-infinite strip. Each lump has a tail 
escaping to infinity, along which the energy density decays like $s^{-\log s}$,
where $s$ is a length variable along the strip. All these solutions are
at least marginally stable, and when $k=0$, are global energy minimizers
in their homotopy class. 

Replacing the potential term by $U=(1+\phi_3)^\alpha$, $\frac12\leq\alpha<1$
we have given a geometric interpretation to the construction of compactons
proposed in \cite{ada1}, and clarified the conditions under which these are
$C^2$ (hence classical solutions of the field equation). In the case of
two-vacuum potentials $U=(1+\phi_3)^\alpha(1-\phi_3)^\beta$, we have constructed
annular compactons, and described how these can be embedded inside one 
another, and deformed into closed string-like solutions.

It is interesting to compare this situation with the case where
$M$ is compact. The role of the potential term $\frac12 U^2$ 
on $\R^2$ is to prevent
lumps dissipating by spreading indefinitely. 
On compact $M$, the very compactness of $M$ 
does this job, so one might
expect that similar results (existence of minimizers in every
homotopy class) might hold here in the simple case $U=0$. This turns out
to be entirely false. Indeed, it was shown in \cite{spesve2} that all
critical points of $E_4(\phi)$ on a compact Riemann surface have 
$\phi^*\omega$
coclosed. Now $\phi^*\omega$ is automatically closed for all $\phi$
(since $\d\phi^*\omega=\phi^*\d\omega=0$), so if $\phi$ solves the
field equation for $E_4$, $\phi^*\omega$ is harmonic. Hence, by
the Hodge Theorem, $\phi^*\omega={\rm constant}\times\vol_M$, that is,
$\phi:M\ra N$ is, up to a homothety of $(M,g)$, an area-preserving 
covering map, or $E_4(\phi)=0$.
 So if the target is $N=S^2$, any solution either has degree $0$,
or is an area-preserving
diffeomorphism $M\ra S^2$
(since $S^2$ is simply connected, any covering map
is a diffeomorphism). It follows that
 if $M=S^2$, the model has solutions only in the 
degree $-1,0,1$ classes, while if $M$ is any other compact Riemann surface,
it has only trivial (degree $0$, energy $0$) solutions. The contrast
with $M=\R^2$ and $U\neq0$ is striking.

The results of this paper raise two obvious interesting questions.
First, can one understand the {\em moduli space} of degree $1$ energy
minimizers of this model? What about the {\em reduced} moduli space,
that is, the set of minimizers modulo
the action of the group of area-preserving diffeomorphisms of $\R^2$?
Clearly, the radially symmetric lump $\psi$, the half lump
$\phi_+$ and the semi-compacton $\phi_\sqsubset$ are three different points
in this space. Do they lie in the same connected component? Is the moduli 
space, in fact, connected? If so, can it be given a manifold structure?
If not, can its components be enumerated? Such questions are mathematically
well-defined (for example, we can give the set of all maps the compact-open
topology, the moduli space the relative topology from this, and
the reduced moduli space the quotient topology from this) but seem formidably
challenging.

Second, can one study the {\em dynamics} of semi-compactons? This
question is rather subtle, because the Euler-Lagrange
equation descending from the obvious Lorentz-invariant
time-dependent extension of the model, with Lagrangian density
\beq\label{ms}
\lll =\frac14[\phi\cdot(\cd_\mu\phi\times\cd_\nu\phi)][\phi\cdot(\cd^\mu\phi
\times\cd^\nu\phi)]-\frac12U(\phi)^2
\eeq
is not a true evolution equation. The problem is that, at any
spatial point $(x,y)\in\R^2$ where $\phi_x$, $\phi_y$ do not span
$T_\phi S^2$ (that is, at any critical point of $\phi(t,\cdot):\R^2\ra S^2$),
the fields $\phi$ and  $\phi_t$, do not uniquely determine $\phi_{tt}$.
In particular, the Cauchy problem for any initial data $\phi(0)$, $\phi_t(0)$
is ill-defined if $\phi(0)$ has any critical points. This is immediately
a problem for any initial data of degree $\geq 2$, since any such field has
critical points by topological considerations. For semi-compactons, the
problem is particularly severe, since these are critical on unbounded regions
of $\R^2$. If the moduli space of semi-compactons can be understood,
one could perhaps study the dynamics of a single semi-compacton within the
geodesic approximation. There are some indications that the kinetic
energy functional of (\ref{ms}) equips the moduli space, at least formally,
with an incomplete Riemannian metric. Less speculatively, one could
abandon Lorentz invariance (which is, in any case, an unnatural assumption
for condensed matter applications) and give the model the
usual kinetic energy term, that is,
\beq
\lll=\frac12\phi_t\cdot\phi_t-\frac12[\phi\cdot(\phi_x\times\phi_y)]^2-\frac12
U(\phi)^2.
\eeq
The Euler-Langrange equation is now a genuine evolution equation, although
it is not technically hyperbolic.
It would be interesting, and 
numerically straightforward,
 to study the scattering of semi-compactons in this
model.

\section*{Acknowledgements}
This work was partially funded by the UK Engineering and Physical Sciences
Research Council.

\end{document}